\documentclass[a4paper,11pt,reqno]{amsart}

\usepackage[margin=35mm]{geometry}
\usepackage[T1]{fontenc}
\usepackage[english]{babel}
\usepackage[utf8]{inputenc}
\usepackage{amsmath,amssymb,amsthm,amsfonts,amstext,amsopn,amsxtra,dsfont,mathrsfs,esint,enumitem,bookmark,mathtools}
\usepackage{bbm}
\usepackage[capitalise]{cleveref}
\hypersetup{linktocpage=true}

\newcommand\numberthis{\addtocounter{equation}{1}\tag{\theequation}}

\theoremstyle{plain}
\newtheorem{theorem}{Theorem}
\newtheorem{lemma}[theorem]{Lemma}

\theoremstyle{definition}

\newcommand\xqed[1]{%
	\leavevmode\unskip\penalty9999 \hbox{}\nobreak\hfill\quad\hbox{#1}%
}
\newcommand\remarkend{\xqed{$\triangle$}}

\makeatletter
\def\@endtheorem{\remarkend\endtrivlist\@endpefalse }
\makeatother
\theoremstyle{remark}

\makeatletter
\def\@endtheorem{\endtrivlist\@endpefalse }
\makeatother

\crefname{theorem}{Theorem}{Theorems}
\crefname{lemma}{Lemma}{Lemmas}
\crefname{proposition}{Proposition}{Propositions}
\crefname{corollary}{Corollary}{Corollaries}
\crefname{definition}{Definition}{Definitions}
\crefname{assumption}{Assumption}{Assumptions}
\crefname{remark}{Remark}{Remarks}

\crefname{subsection}{subsection}{subsections}
\crefname{subsubsection}{subsection}{subsections}

\renewcommand{\d}[1]{\ensuremath{\operatorname{d}\!{#1}}}

\newcommand{\diver}{\operatorname{div}}

\newcommand\mydots{\ifmmode\mathellipsis\else.\kern-0.08em.\kern-0.08em.\fi}

\allowdisplaybreaks

\DeclareUnicodeCharacter{2217}{*}

\usepackage{csquotes}
\usepackage[backend=bibtex,sorting=nty,citestyle=numeric-comp,bibstyle=ieee,maxbibnames=99,doi=false,isbn=false,url=false,eprint=false,dashed=true]{biblatex}
\AtEveryBibitem{\clearfield{month}}
\AtEveryBibitem{\clearfield{day}}
\usepackage{graphicx}
\usepackage{subcaption}
\usepackage[table,xcdraw]{xcolor}

\usepackage{tikz}
\usetikzlibrary{arrows}
\usepackage{tkz-tab}

\usepackage{standalone}

\usepackage{hyperref}
\hypersetup{bookmarksdepth=3}
\title[Ground state energy of a three-body hard-core Bose gas]{Ground state energy of a dilute Bose gas with three-body hard-core interactions}
\author[L. Junge]{Lukas Junge}
\address{Department of Mathematics, Copenhagen university, Lyngbyvej 2, 2100 Copenhagen, Denmark}
\email{lj@math.ku.dk}
\author[F. L. A. Visconti]{François L. A. Visconti}
\address{Department of Mathematics, LMU Munich, Theresienstrasse 39, 80333 Munich, Germany}
\email{visconti@math.lmu.de}

\begin{document}
	\maketitle
	
	\section*{Abstract}
	We consider a gas of bosons interacting through a three-body hard-core potential in the thermodynamic limit. We derive an upper bound on the ground state energy of the system at the leading order using a Jastrow factor. Our result matches the lower bound proven by Nam--Ricaud--Triay \cite{Nam2022ground} and therefore resolves the leading order. Moreover, a straightforward adaptation of our proof can be used for systems interacting via combined two-body and three-body interactions to generalise \cite[Theorem 1.2]{Visconti26} to hard-core potentials.
		
	\tableofcontents
	
	\section{Introduction}
	A system of $N$ bosons trapped in a box $\Lambda_L \coloneqq \left[0,L\right]^3$ interacting via three-body interactions can be described by the Hamiltonian operator
	\begin{equation}
		\label{eq:general_three_body_hamiltonian}
		H_{N,L} = \sum_{i=1}^N-\Delta_i + \sum_{1\leq i<j<k\leq N}w(x_i-x_j,x_i-x_k)
	\end{equation}
	acting on the Hilbert space $L_\textmd{s}^2(\Lambda_L^N)$ - the subspace of $L^2(\Lambda_L^N)$ consisting of functions that are symmetric with respect to permutations of the $N$ particles - and where $-\Delta$ denotes the Laplacian with Neumann boundary conditions on $\Lambda_L$ ($\nabla f \cdot \vec{n} = 0 \textmd{ on } \partial \Lambda_L$). Such systems have received a lot of attention in recent years and have been the subject of many mathematical works \cite{AdamiLee25,Chen2011quinticNLS,Chen11,Chen2018TheDO,lee2021rateCT,Li2021derivationNS,NamSal19,Nguyen2023onedimensional,NguyenRic25,RoutSoh23,Xie15,Yuan2015derivationQNLS}.
	
	In \cite{Nam2022ground}, Nam--Ricaud--Triay proved that for a nonnegative, compactly supported potential $w\in L^\infty(\mathbb{R}^6)$, the  Hamiltonian \eqref{eq:general_three_body_hamiltonian} satisfies
	\begin{equation}
		\label{eq:general_three_body_hamiltonian_gse_upper_bound}
		\quad \lim_{\substack{N,L\rightarrow\infty\\ N/L^3\rightarrow\rho}}\dfrac{\inf\sigma(H_{N,L})}{N} = \dfrac{1}{6}\rho^2b_{\mathcal{M}}(w)(1 + O((\rho b_\mathcal{M}(w)^{3/4})^\nu))
	\end{equation}
	in the dilute regime $\rho b_\mathcal{M}(w)^{3/4} \rightarrow 0$, for some constant $\nu > 0$. Here, $\rho$ denotes the density of the system and $b_\mathcal{M}(w)$ is the scattering energy associated to $w$, which characterises the energy of three-particle collisions and is defined in \eqref{eq:three_body_scattering_energy}. This result was then improved in \cite{Visconti26}, where it was shown that \eqref{eq:general_three_body_hamiltonian_gse_upper_bound} holds for $w\geq 0$ compactly supported and satisfying
	\begin{equation*}
		\|w\|_{L^2L^1} \coloneqq \Big(\int_{\mathbb{R}^3}\|w(x,\cdot)\|_{L^1(\mathbb{R}^3)}^2\d{}x\Big)^{1/2} < \infty.
	\end{equation*}
	It was also shown in \cite{Visconti26} that \eqref{eq:general_three_body_hamiltonian_gse_upper_bound} holds with an error in $o(1)$ for $w$ of class $L^1$. The goal of this paper is to prove that \eqref{eq:general_three_body_hamiltonian_gse_upper_bound} remains valid for particles interacting with a hard-core potential.
	
	We consider a gas of $N$ bosons with three-body hard-core interactions in $\Lambda_L = [0,L]^3$. We study the ground state energy
	\begin{equation}
		\label{eq:ground_state_energy_definition}
		E_{N,L} = \inf\dfrac{\left<\Psi,\sum_{i=1}^N-\Delta_i\Psi\right>}{\|\Psi\|^2},
	\end{equation}
	with the infimum taken over all $\Psi\in L_\textmd{s}^2(\Lambda_L^N)$ satisfying the three-body hard-core condition $\Psi(x_1,\dots,x_N) = 0$ if there exist $i,j,k\in\{1,\dots,N\}$, $i\neq j\neq k\neq i$ with $|(x_i-x_j,x_i-x_k,x_j-x_k)|/\sqrt{3} \leq \mathfrak{a}$. \footnote{Though there is no canonical choice for the three-body hard-core potential, the present choice is motivated by the physics literature (see e.g. \cite{Comparin2015Liquid,Piatecki2014Efimov}).} Here, $|\cdot|$ is the Euclidean norm in $\mathbb{R}^9$ and $\sum_{i=1}^N-\Delta_i$ is to be understood in the quadratic form sense in $H^1_\textmd{s}(\Lambda_L^N)$. The three-body scattering energy associated to the hard-core potential
	\begin{equation*} 
		w_\textmd{hc}(x-y,x-z) =
		\left\{
		\begin{array}{ll}
			+\infty & \textmd{if $|(x-y,x-z,y-z)|/\sqrt{3} \leq \mathfrak{a}$,}\\
			0 & \textmd{otherwise}
		\end{array}
		\right.
	\end{equation*}
	is given by
	\begin{equation}
		\label{eq:hard_core_potential_scattering_energy}
		b_\mathcal{M}(w_\textmd{hc}) = \dfrac{64}{3\sqrt{3}}\pi^2\mathfrak{a}^4.
	\end{equation}
	This is shown in the proof of Lemma \ref{lemma:scattering_properties}.
	
	Here is our main result.
	
	\begin{theorem}
		\label{th:energy_upper_bound}
		There exist a constant $C > 0$ (independent of $\mathfrak{a}$ and $\rho$) such that
		\begin{equation}
			\label{eq:energy_upper_bound}
			\lim_{\substack{N,L\rightarrow\infty\\ N/L^3\rightarrow\rho}}\dfrac{E_{N,L}}{N} = \dfrac{32}{9\sqrt{3}}\pi^2\rho^2\mathfrak{a}^4\left(1 + C\left(\rho\mathfrak{a}^3\right)^{4/7}\right)
		\end{equation}
		for all $\rho\mathfrak{a}^3$ small enough.
	\end{theorem}
	
	In the present paper, we prove an upper bound on \eqref{eq:energy_upper_bound}, which matches the lower bound proven in \cite{Nam2022ground}. Here are some remarks on the result:
	\begin{enumerate}
		\item The strategy of the proof of Theorem~\ref{th:energy_upper_bound} can be used to generalise \eqref{eq:general_three_body_hamiltonian_gse_upper_bound} to $w$ of class $L^1$ with an error uniform in $w$ assuming that $R_0/b_\mathcal{M}(w)^{1/4}$ remains bounded, where $R_0$ is the range of $w$.
		\item In \cite[Conjecture 8]{Nam2022dilute}, a heuristic approach was used to predict that the ground state energy of a system described by the Hamiltonian \eqref{eq:general_three_body_hamiltonian} should satisfy
		\begin{equation}
			\label{eq:energy_second_order_conjecture}
			\quad \lim_{\substack{N,L\rightarrow\infty\\ N/L^3\rightarrow\rho}}\dfrac{\inf\sigma(H_{N,L})}{N} = \dfrac{1}{6}\rho^2b_{\mathcal{M}}(w)(1 + C(w)\rho + o(\rho))
		\end{equation}
		in the low density regime $\rho \rightarrow 0$, for some constant $C(w)$ depending only on $w$. The error yielded by the proof of Theorem~\ref{th:energy_upper_bound} is of order $(\rho\mathfrak{a}^3)^{4/7} \gg \rho\mathfrak{a}^3$, meaning that it does not capture the error at the correct order. The same problem arises in the two-body case when only using cancellations between the numerator and the denominator similar to the ones used in \eqref{eq:energy_bound_I1}--\eqref{eq:energy_bound_I3} (see for example \cite{Basti2022gseGP}). To extract an error at the correct order in the two-body case one needs to push the analysis much further and identify additional cancellations, as was done in \cite{BastiCenGiuOlg24}.
		\item A straightforward adaptation of the  proof of Theorem~\ref{th:energy_upper_bound} can be used to derive a correct upper bound at the first order of the ground state energy of a system interacting via combined two-body and three-body interactions. More specifically, the ground state energy $E_{N,L}'$ of a system of $N$ bosons located in $\Lambda_L$ interacting via a two-body hard-core potential of radius $\mathfrak{a}_\textmd{2}$ and a three-body hard-core potential of radius $\mathfrak{a}_\textmd{3} > \mathfrak{a}_\textmd{2}$  is such that
		\begin{equation*}
			\qquad\quad \lim_{\substack{N,L\rightarrow\infty\\ N/L^3\rightarrow\rho}}\dfrac{E_{N,L}'}{N} = \Big(4\pi\rho\mathfrak{a}_\textmd{2} + \dfrac{32}{9\sqrt{3}}\pi^2\rho^2\mathfrak{a}_\textmd{3}^4\Big) \left(1 + C\left(\rho\mathfrak{a}_\textmd{2}^3\right)^{\nu_1} + C\left(\rho\mathfrak{a}_\textmd{3}^3\right)^{\nu_2}\right)
		\end{equation*}
		for all $\rho\mathfrak{a}_2^3$ and $\rho\mathfrak{a}_3^3$ small enough and for some $\nu_1,\nu_2 >0$ (independent of $\mathfrak{a}_2, \mathfrak{a}_3$ and $\rho$). This generalises \cite[Theorem 1.2]{Visconti26} to hard-core potentials. The lower bound proven in \cite{Visconti26} easily generalises to the hard-core case, whereas the proof of the upper bound, which relies on second quantisation techniques, fails. Instead, one proves the upper bound using a trial state of the form
		\begin{equation*}
			\qquad \Psi(x_1,\dots,x_N) = \prod_{1\leq i<j\leq N}\widetilde{f}_{\ell_1}(x_i - x_j)\prod_{1\leq i<j<k\leq N}f_{\ell_2}(x_i - x_j,x_i - x_k),
		\end{equation*}
		where $\widetilde{f}_{\ell_1}$ describes the two-body correlations up to a distance $\ell_1$ and $f_{\ell_2}$ describes the three-body correlations up to a distance $\ell_2$. Note that this trial state is simply the product of the trial considered in the proof of Theorem~\ref{th:energy_upper_bound} (see \eqref{eq:trial_state}) and the standard trial state used in the two-body case.
	\end{enumerate}
		
	\section{Scattering properties of the three-body hard-core potential}
	Since we are considering a dilute gas the correlation structure is encoded in the zero-energy scattering problem
	\begin{equation}
		\label{eq:zero_scattering_equation}
		(-\Delta_x -\Delta_y -\Delta_z)f(x-y,x-z) = 0
	\end{equation}
	on $\mathbb{R}^9$, where $f$ satisfies the conditions $f(x-y,x-z) = 0$ if $|(x-y,x-z,y-z)|/\sqrt{3} \leq \mathfrak{a}$ and $f(\mathbf{x}) \rightarrow 1$ as $|\mathbf{x}| \rightarrow \infty$. Note that $f$ satisfies the three-body symmetry properties
	\begin{equation}
		\label{eq:three_body_symmetry}
		f(x,y) = f(y,x) \quad \textmd{and} \quad f(x - y,x - z) = f(y - x,y - z) = f(z - x, z - y),
	\end{equation}
	for all $x,y,z\in\mathbb{R}^3$.
	
	By removing the centre of mass using the change of variables
	\begin{equation}
		\label{eq:removal_centre_of_mass_change_of_variables}
		r_1 = \dfrac{1}{3}(x + y + z), \quad r_2 = x-y, \quad \textmd{and} \quad r_3 = x-z,
	\end{equation}
	we find that the scattering problem \eqref{eq:zero_scattering_equation} is equivalent to the modified zero-energy scattering problem
	\begin{equation}
		\label{eq:modified_zero_scattering_equation}
		-2\Delta_\mathcal{M}f(r_2,r_3) = 0
	\end{equation}
	on $\mathbb{R}^6$, with $f$ satisfying the conditions $f(r_2,r_3) = 0$ if $\left|\mathcal{M}^{-1}(r_2,r_3)\right| \leq \sqrt{2}\mathfrak{a}$ and $f(\mathbf{x}) \rightarrow 1$ as $|\mathbf{x}| \rightarrow \infty$. Here, we introduced the modified Laplacian
	\begin{equation*}
		-\Delta_\mathcal{M} = -|\mathcal{M}\nabla_{\mathbb{R}^6}|^2 = -\diver_{\mathbb{R}^6}(\mathcal{M}^2\nabla_{\mathbb{R}^6}),
	\end{equation*}
	where the matrix $\mathcal{M}:\mathbb{R}^3\times\mathbb{R}^3 \rightarrow \mathbb{R}^3\times\mathbb{R}^3$ is given by
	\begin{equation*}
		\label{eq:modified_matrix}
		\mathcal{M} \coloneqq \left(
		\dfrac{1}{2}
		\begin{pmatrix}
			2 & 1\\
			1 & 2
		\end{pmatrix}
		\right)^{1/2}
		= \dfrac{1}{2\sqrt{2}}
		\begin{pmatrix}
			\sqrt{3} + 1 & \sqrt{3} - 1\\
			\sqrt{3} - 1 & \sqrt{3} + 1
		\end{pmatrix},
	\end{equation*}
	with inverse
	\begin{equation*}
		\mathcal{M}^{-1} = \left(
		\dfrac{2}{3}
		\begin{pmatrix}
			2 & -1\\
			-1 & 2
		\end{pmatrix}
		\right)^{1/2}
		= \dfrac{1}{\sqrt{6}}
		\begin{pmatrix}
			1 + \sqrt{3} & 1 - \sqrt{3}\\
			1 - \sqrt{3} & 1 + \sqrt{3}
		\end{pmatrix}.
	\end{equation*}
	Note that $\det\mathcal{M} = \sqrt{3}/2$. Define the scattering energy $b_\mathcal{M}(w)$ of some three-body potential $w$ by
	\begin{equation}
		\label{eq:three_body_scattering_energy}
		b_\mathcal{M}(w) = \inf_{g\in D^1(\mathbb{R}^6)}\int_{\mathbb{R}^6}\d{}\mathbf{x}\left(2\vert\mathcal{M}\nabla g(\mathbf{x})\vert^2 + w(\mathbf{x})\vert 1 - g(\mathbf{x})\vert^2\right),
	\end{equation}
	where $D^1(\mathbb{R}^6)$ denotes the space of functions $g:\mathbb{R}^6 \rightarrow\mathbb{C}$ in $L^1_{\textmd{loc}}(\mathbb{R}^6)$ vanishing at infinity and such that $\vert\nabla g\vert\in L^2(\mathbb{R}^6)$ (see \cite[Section 8.3]{lieb_loss}). We refer to \cite{Nam2023condensation} for a more in-depth discussion on the matter.
	
	Let $f$ denote a solution to \eqref{eq:modified_zero_scattering_equation} and define $\tilde{f} = f(\mathcal{M}\cdot)$. Then, $\tilde{f}$ solves
	\begin{equation*}
		\label{eq:zero_scattering_equation_radial}
		-\Delta \tilde{f} = 0
	\end{equation*}
	on $\mathbb{R}^6$, with the conditions $\tilde{f}(\mathbf{x}) = 0$ for $|\mathbf{x}| \leq \sqrt{2}\mathfrak{a}$ and $\tilde{f}(\mathbf{x}) \rightarrow 1$ as $|\mathbf{x}| \rightarrow \infty$. By rewriting the previous problem in hyperspherical coordinates we find that \eqref{eq:modified_zero_scattering_equation} has for unique solution
	\begin{equation*}
		\label{eq:modified_scattering_solution}
		f(\mathbf{x}) = 
		\left\{
		\begin{array}{ll}
			1 - \dfrac{4\mathfrak{a}^4}{|\mathcal{M}^{-1}\mathbf{x}|^4} \quad &\textmd{if $|\mathcal{M}^{-1}\mathbf{x}| > \sqrt{2}\mathfrak{a}$},\\
			0 &\textmd{otherwise}.
		\end{array}
		\right.
	\end{equation*}
	Define also
	\begin{equation*}
		\omega(\mathbf{x}) = 1 - f(\mathbf{x}),
	\end{equation*}
	for all $\mathbf{x}\in\mathbb{R}^6$.
	
	\sloppy We shall need a truncated version of $f$ with a cut-off. Let $\tilde{\chi}\in C^\infty(\mathbb{R}^6;[0,1])$ be a radial function satisfying $\tilde{\chi}(\mathbf{x}) = 1$ if $|\mathbf{x}|\leq 1/2$ and $\tilde{\chi}(\mathbf{x}) = 0$ if $|\mathbf{x}| \geq 1$, and define $\chi = \tilde{\chi}(\mathcal{M}^{-1}\cdot)$. For all $\ell\in(\mathfrak{a},L)$, we define
	\begin{equation}
		\label{eq:truncated_modified_scattering_solution}
		\chi_\ell = \chi(\ell^{-1}\cdot), \quad \omega_\ell = \chi_\ell\omega, \quad \textmd{and} \quad f_\ell = 1 - \omega_\ell.
	\end{equation}
	Note that $\omega, f, \omega_\ell$ and $f_\ell$ satisfy the three-body symmetry \eqref{eq:three_body_symmetry}. Moreover, they have the following properties:
	\begin{lemma}
		\label{lemma:scattering_properties}
		Let $\ell \in(\mathfrak{a},L)$. Then,
		\begin{equation}
			\label{eq:truncated_modified_scattering_solution_estimate_derivative}
			|\nabla f_\ell(\mathbf{x})| \leq C\mathfrak{a}^4\dfrac{\mathds{1}_{\left\{C_1\mathfrak{a}\leq |\mathbf{x}|\leq C_2\ell\right\}}}{|\mathbf{x}|^5}
		\end{equation}
		and
		\begin{equation}
			\label{eq:truncated_modified_scattering_solution_estimate_u}
			0 \leq 1 - f_\ell^2(\mathbf{x}) \leq C\mathfrak{a}^4\dfrac{\mathds{1}_{\left\{C_1\mathfrak{a}\leq |\mathbf{x}|\leq C_2\ell\right\}}}{|\mathbf{x}|^4},
		\end{equation}
		for all $\mathbf{x}\in\mathbb{R}^6$. Here, $C,C_1,C_2$ are universal positive constants such that $C_1 < C_2$. Moreover,
		\begin{equation}
			\label{eq:truncated_modified_scattering_solution_energy}
			\int_{\mathbb{R}^6}\d{}\mathbf{x}\left(|(\mathcal{M}\nabla_{\mathbb{R}^6}f_\ell)(\mathbf{x})|^2\right) \leq \dfrac{32}{3\sqrt{3}}\pi^2\mathfrak{a}^4\Big(1 + C\Big(\dfrac{\mathfrak{a}}{\ell}\Big)^4\Big)
		\end{equation}
		Furthermore, define $\tilde{\ell} = \sqrt{3/2}\ell$ and
		\begin{equation}
			\label{eq:truncated_modified_scattering_solution_estimate_g}
			g_\ell(x) = \mathds{1}(|x| \geq \tilde{\ell}).
		\end{equation}
		Then,
		\begin{equation}
			\label{eq:truncated_modified_scattering_solution_estimate_g_inequality}
			f_\ell(x_1,x_2) \geq \max(g_\ell(x_1),g_\ell(x_2)),
		\end{equation}
		for all $x_1,x_2\in\mathbb{R}^3$.
	\end{lemma}
	
	\begin{proof}
		Both \eqref{eq:truncated_modified_scattering_solution_estimate_derivative} and \eqref{eq:truncated_modified_scattering_solution_estimate_u} follow directly from the definition of $f_\ell$ and $|\nabla\chi_\ell(\mathbf{x})| \leq C\ell^{-1}\mathds{1}_{\left\{\ell/2\leq |\mathcal{M}^{-1}\mathbf{x}|\leq \ell\right\}}$ and $\sigma(\mathcal{M}^{-1}) = \left\{\sqrt{2/3},\sqrt{2}\right\}$. To compute \eqref{eq:truncated_modified_scattering_solution_energy} we first expand
		\begin{align*}
			\int_{\mathbb{R}^6}\d{}\mathbf{x}\Big(|(\mathcal{M}\nabla_{\mathbb{R}^6}f_\ell)(\mathbf{x})|^2\Big) &= \int_{\mathbb{R}^6}\d{}\mathbf{x}\Big(|(\mathcal{M}\nabla_{\mathbb{R}^6}\omega)(\mathbf{x})|^2\chi_\ell(\mathbf{x})^2\Big)\\
			&\phantom{=} + 2\int_{\mathbb{R}^6}\d{}\mathbf{x}\Big((\mathcal{M}\nabla_{\mathbb{R}^6}\omega)(\mathbf{x})\cdot(\mathcal{M}\nabla_{\mathbb{R}^6} \chi_\ell)(\mathbf{x})\omega_\ell(\mathbf{x})\Big)\\
			&\phantom{=} + \int_{\mathbb{R}^6}\d{}\mathbf{x}\Big(\omega(\mathbf{x})^2|(\mathcal{M}\nabla_{\mathbb{R}^6}\chi_\ell)(\mathbf{x})|^2\Big).
		\end{align*}
		The only contribution of order $\mathfrak{a}^4$ comes from the first term. Indeed, using again $|\nabla\chi(\mathbf{x})| \leq C\ell^{-1}\mathds{1}_{\left\{\ell/2\leq |\mathcal{M}^{-1}\mathbf{x}|\leq \ell\right\}}$ and \eqref{eq:truncated_modified_scattering_solution_estimate_derivative} we have
		\begin{equation*}
			\int_{\mathbb{R}^6}\d{}\mathbf{x}\Big(2(\mathcal{M}\nabla_{\mathbb{R}^6}\omega)(\mathbf{x})\cdot(\mathcal{M}\nabla_{\mathbb{R}^6} \chi_\ell)(\mathbf{x})\omega_\ell(\mathbf{x}) + |\omega(\mathbf{x})^2|(\mathcal{M}\nabla_{\mathbb{R}^6}\chi_\ell)(\mathbf{x})|^2\Big)\\
			\leq C\mathfrak{a}^4\left(\dfrac{\mathfrak{a}}{\ell}\right)^4.
		\end{equation*}
		Moreover, by writing $\omega = \tilde{\omega}(\mathcal{M}^{-1}\cdot)$ with $\tilde{\omega}(\mathbf{x}) = 4\mathfrak{a}^4/|\mathbf{x}|^4$ if $\vert\mathbf{x}\vert > \sqrt{2}\mathfrak{a}$ and $\tilde{\omega}(\mathbf{x}) = 1$ otherwise, we get
		\begin{align*}
			\int_{\mathbb{R}^6}\d{}\mathbf{x}\Big(|(\mathcal{M}\nabla_{\mathbb{R}^6}\omega)(\mathbf{x})|^2\chi_\ell(\mathbf{x})^2\Big) &= \int_{\mathbb{R}^6}\d{}\mathbf{x}\Big(\left|(\nabla_{\mathbb{R}^6}\tilde{\omega})(\mathcal{M}^{-1}\mathbf{x})\right|^2\tilde{\chi}_\ell(\mathcal{M}^{-1}\mathbf{x})^2\Big)\\
			&= \det\mathcal{M}\int_{\mathbb{R}^6}\d{}\mathbf{y}\Big(\left|(\nabla_{\mathbb{R}^6}\tilde{\omega})(\mathbf{y})\right|^2\tilde{\chi}_\ell(\mathbf{y})^2\Big)\\
			&\leq \dfrac{32}{3\sqrt{3}}\pi^2\mathfrak{a}^4.
		\end{align*}
		In the second equality we used the change of variables $\mathbf{y} = \mathcal{M}^{-1}\mathbf{x}$. In the last inequality we used $\nabla_{\mathbb{R}^6}\tilde{\omega} = 0$ on $B(0,\sqrt{2}\mathfrak{a})$ and $\tilde{\chi}_\ell(\mathbf{x}) \leq \mathds{1}_{\left\{|\mathbf{x}| \leq \ell\right\}}$ and $\nabla_\mathbf{y}(1/|\mathbf{y}|^4) = -4\mathbf{y}/|\mathbf{y}|^6$ and $\det\mathcal{M} = \sqrt{3}/2$ and that the surface of the $5$-dimensional sphere in $\mathbb{R}^6$ is given by $|\mathbb{S}^5| = 8\pi^2/3$. This proves \eqref{eq:truncated_modified_scattering_solution_energy}. Notice also that the same computation gives \eqref{eq:hard_core_potential_scattering_energy}.
		
		Finally, notice that $f_\ell(x_1,x_2) = 1$ when $|\mathcal{M}^{-1}(x_1,x_2)|^{-1}\geq \ell$, which is true whenever $|x_1| \geq \tilde{\ell}$ or $|x_2| \geq \tilde{\ell}$. This immediately implies \eqref{eq:truncated_modified_scattering_solution_estimate_g_inequality} and concludes the proof of Lemma~\ref{lemma:scattering_properties}.
	\end{proof}
		
	\section{Proof of the upper bound}
	To get an upper bound on \eqref{eq:ground_state_energy_definition}, we need to evaluate the energy on an appropriate trial state. To do so, we add correlations among particles to the uncorrelated state $\tilde{\Psi}_{N,L} \equiv 1$. Since correlations are produced mainly by three-body scattering events, we consider the trial state
	\begin{equation}
		\label{eq:trial_state}
		\Psi_{N,L}(x_1,\dots,x_N) = \prod_{1\leq i<j<k\leq N}f_\ell(x_i-x_j,x_i-x_k),
	\end{equation}
	where $\ell$ is a parameter satisfying $\mathfrak{a} \ll \ell \ll L$ that will be fixed later; the state $\Psi_{N,L}$ clearly satisfies the three-body hard-core condition imposed in \eqref{eq:ground_state_energy_definition}. The function $f_\ell$ defined in \eqref{eq:truncated_modified_scattering_solution} describes the three-body correlations up to a distance $\ell$. Trial states of this form have been first used in \cite{Bijl1940lowestWF,Dingle1949ZPE,Jastrwo1955ManyBody} in the two-body case and are usually referred to as Jastrow factors. Note that though Dyson \cite{Dyson1957gseHS} worked with a nonsymmetric trial state describing only nearest-neighbour correlations, trial states similar to \eqref{eq:trial_state} can be used to derive the leading order of a Bose gas with two-body interactions, as shown in \cite{Basti2022gseGP}. For readability's sake, we from now on write
	\begin{equation*}
		f_{ijk} = f_\ell(x_i-x_j,x_i-x_k)
	\end{equation*}
	and
	\begin{equation}
		\nabla_if_{ijk} = \nabla_{x_i}f_\ell(x_i-x_j,x_i-x_k)
	\end{equation}
	for all $i,j,k\in\{1,\dots,N\}$.
	
	To compute the energy of the trial state \eqref{eq:trial_state}, we first notice that
	\begin{equation*}
		\nabla_{x_1}\Psi_{N,L}(x_1,\dots,x_N) = \sum_{2\leq p<q \leq N}\dfrac{\nabla_{1}f_{1pq}}{f_{1pq}}\prod_{1\leq i<j<k\leq N}f_{ijk},
	\end{equation*}
	which when combined with the three-body symmetry \eqref{eq:three_body_symmetry} implies
	\begin{multline*}
		\dfrac{\left<\Psi_{N,L},\sum_{i=1}^N-\Delta_i\Psi_{N,L}\right>}{\|\Psi_{N,L}\|^2} = N\dfrac{\left<\nabla_1\Psi_{N,L},\nabla_1\Psi_{N,L}\right>}{\|\Psi_{N,L}\|^2}\\
		\begin{aligned}[t]
			&= \dfrac{N(N-1)(N-2)}{3}\dfrac{\int\d{}\mathbf{x}_N\left(\frac{|\mathcal{M}\nabla f_{123}|^2}{f_{123}^2}\prod_{1\leq i<j<k\leq N}f_{ijk}^2\right)}{\int\d{}\mathbf{x}_N\left(\prod_{1\leq i<j<k\leq N}f_{ijk}^2\right)}\\
			&\phantom{=} + N(N-1)(N-2)(N-3)\dfrac{\int\d{}\mathbf{x}_N\left(\frac{\nabla_1f_{123}}{f_{123}}\cdot\frac{\nabla_1f_{124}}{f_{124}}\prod_{1\leq i<j<k\leq N}f_{ijk}^2\right)}{\int\d{}\mathbf{x}_N\left(\prod_{1\leq i<j<k\leq N}f_{ijk}^2\right)}\\
			&\phantom{=} + \dfrac{N(N-1)(N-2)(N-3)(N-4)}{4}\dfrac{\int\d{}\mathbf{x}_N\left(\frac{\nabla_1f_{123}}{f_{123}}\cdot\frac{\nabla_1f_{145}}{f_{145}}\prod_{1\leq i<j<k\leq N}f_{ijk}^2\right)}{\int\d{}\mathbf{x}_N\left(\prod_{1\leq i<j<k\leq N}f_{ijk}^2\right)}\\
			&\eqqcolon \mathcal{I}_1 + \mathcal{I}_2 + \mathcal{I}_3.
		\end{aligned}
	\end{multline*}
	In the second equality we used
	\begin{multline*}
		\left|\nabla_{x_1}f_\ell(x_1 - x_2,x_1 - x_3)\right|^2 + \left|\nabla_{x_2}f_\ell(x_1 - x_2,x_1 - x_3)\right|^2\\
		+
		\left|\nabla_{x_3}f_\ell(x_1 - x_2,x_1 - x_3)\right|^2 = 2\left|(\mathcal{M}\nabla_{\mathbb{R}^6} f_\ell)(x_1 - x_2,x_1 - x_3)\right|^2.
	\end{multline*}
	Let us now bound each term one by one. Thanks to \eqref{eq:truncated_modified_scattering_solution_estimate_g_inequality}, we have
	\begin{equation}
		\label{eq:jastrow_factor_lower_bound_less_correlations}
		\prod_{3\leq j<k\leq N}f_\ell(x_i-x_j,x_i-x_k)^2 \geq \prod_{j=3}^Ng_\ell(x_i - x_j),
	\end{equation}
	for $i\in\{1,2\}$. Hence, by defining $u_\ell = 1 - f_\ell^2$ and $v_\ell = 1 - g_\ell$ we have the estimate
	\begin{equation}
		\label{eq:jastrow_factor_almost_one}
		1 - \sum_{j=3}^Nv_{1j} - \sum_{j=3}^Nv_{2j} - \sum_{k = 3}^Nu_{12k} \leq \prod_{3\leq j<k\leq N}f_{1jk}^2f_{2jk}^2\prod_{k=3}^Nf_{12k}^2 \leq 1,
	\end{equation}
	where we used the short-hand notations $v_{ij} = v_\ell(x_i - x_j)$ and $u_{ijk} = u_\ell(x_i-x_j,x_i-x_k)$. Notice that, thanks to \eqref{eq:jastrow_factor_lower_bound_less_correlations}, the left-hand side of \eqref{eq:jastrow_factor_almost_one} contains solely single sums, and no double sums as it would have if we did not use this bound; this plays an important role in the estimates below. The bound \eqref{eq:jastrow_factor_almost_one} allows us to decouple the variables $x_1$ and $x_2$ in the numerator and in the denominator of $\mathcal{I}_1$; with \eqref{eq:truncated_modified_scattering_solution_estimate_u} and \eqref{eq:truncated_modified_scattering_solution_energy} we obtain
	\begin{align*}
		\mathcal{I}_1 &\leq \dfrac{N^3}{3}\dfrac{\int_{\mathbb{R}^6}\d{}\mathbf{x}\left(|(\mathcal{M}\nabla_{\mathbb{R}^6}f_\ell)(\mathbf{x})|^2\right)}{L^6 - CL^3N\int\d{}x(v_\ell(x)) - CN\int\d{}\mathbf{x}(u_\ell(\mathbf{x}))}\\
		&\leq \dfrac{32}{9\sqrt{3}}\pi^2\dfrac{N\rho^2\mathfrak{a}^4(1 + C(\mathfrak{a}/\ell)^4)}{1 - C\rho\ell^3 - C\rho \ell^2\mathfrak{a}^4/L^3}\\
		&\leq \dfrac{32}{9\sqrt{3}}\pi^2N\rho^2\mathfrak{a}^4\Big(1 + C\Big(\dfrac{\mathfrak{a}}{\ell}\Big)^4 + C\rho\ell^3\Big), \numberthis \label{eq:energy_bound_I1}
	\end{align*}
	under the assumption that $\rho\ell^3 \ll 1$ and $\mathfrak{a}\ll \ell \ll L$.  In the last inequality we used $\ell^2\mathfrak{a}^4/L^3 \leq \ell^3$. To bound $\mathcal{I}_2$ we similarly decouple the variables $x_2, x_3$ and $x_4$. Using again \eqref{eq:truncated_modified_scattering_solution_estimate_derivative} and \eqref{eq:truncated_modified_scattering_solution_estimate_u} we can bound
	\begin{align*}
		\mathcal{I}_2 &\leq CN^4\dfrac{\int\d{}x\d{}y\d{}z(|\nabla f_\ell(x,y)|\cdot|\nabla f_\ell(x,z)|)}{L^9 - CNL^6\int\d{}x(v_\ell(x)) - CNL^3\int\d{}\mathbf{x}(u_\ell(\mathbf{x}))}\\
		&\leq CN\rho^2\mathfrak{a}^4\left[\rho\mathfrak{a}^4\ell\right]\left(1 + C\rho\ell^3\right)\\
		&\leq CN\rho^2\mathfrak{a}^4\left(\rho\ell^3\right), \numberthis \label{eq:energy_bound_I2}
	\end{align*}
	when $\rho\ell^3 \ll 1$ and $\mathfrak{a}\ll \ell \ll L$.
	Analogously, we bound $\mathcal{I}_3$ by decoupling the variables $x_1,x_2$ and $x_4$. Namely, using once more \eqref{eq:truncated_modified_scattering_solution_estimate_derivative} and \eqref{eq:truncated_modified_scattering_solution_estimate_u} we get
	\begin{align*}
		\mathcal{I}_3 &\leq CN^5\dfrac{\left(\int\d{}\mathbf{x}|\nabla f_\ell(\mathbf{x})|\right)^2}{L^{12} - CNL^9\int\d{}x(v_\ell(x)) - CNL^6\int\d{}\mathbf{x}(u_\ell(\mathbf{x}))}\\
		&\leq CN\rho^2\mathfrak{a}^4\left[\rho^2\mathfrak{a}^4\ell^2\right]\left(1 + C\rho\ell^3\right)\\
		&\leq CN\rho^2\mathfrak{a}^4\left(\rho\ell^3\right) \numberthis \label{eq:energy_bound_I3}
	\end{align*}
	again under the condition that $\rho\ell^3 \ll 1$ and $\mathfrak{a}\ll \ell \ll L$.
	From \eqref{eq:energy_bound_I1}--\eqref{eq:energy_bound_I3} we conclude that
	\begin{equation*}
		E_{N,L} \leq \dfrac{32}{9\sqrt{3}}\pi^2N\rho^2\mathfrak{a}^4\Big(1 + C\Big(\dfrac{\mathfrak{a}}{\ell}\Big)^4 + C\rho\ell^3\Big).
	\end{equation*}
	Taking $\ell = \mathfrak{a}\left(\rho\mathfrak{a}^3\right)^{-1/7}$ finishes the proof of Theorem~\ref{th:energy_upper_bound}.
	
	\section*{Acknowledgments.}
	We thank Arnaud Triay and Phan Thành Nam for their precious feedback. L. J. was partially supported by the European Union. We also thank the anonymous referee for the precious feedback. Views and opinions expressed are however those of the authors only and do not necessarily reflect those of the European Union or the European Research Council. Neither the European Union nor the granting authority can be held responsible for them. L. J. was partially supported by the Villum Centre of Excellence for the Mathematics of Quantum Theory (QMATH) with Grant No.10059. L. J. was supported by the grant 0135-00166B from Independent Research Fund Denmark. F. L. A. V. acknowledges partial support by the Deutsche Forschungsgemeinschaft (DFG, German Research Foundation) through the TRR 352 Project ID 470903074 and by the European Research Council through the ERC CoG RAMBAS Project Nr. 101044249. 
	
	\section*{Conflict of interest}
	On behalf of all authors, the corresponding author states that there is no conflict of interest.
	
	\section*{Data availability statement}
	Data sharing not applicable to this article as no datasets were generated or analysed during the current study.
	
	\printbibliography

@Article{Basti2022gseGP,
  author    = {Basti, Giulia and Cenatiempo, Serena and Olgiati, Alessandro and Pasqualetti, Giulio and Schlein, Benjamin},
  journal   = {J. Math. Phys.},
  title     = {Ground state energy of a {B}ose gas in the {G}ross--{P}itaevskii regime},
  year      = {2022},
  issn      = {1089-7658},
  month     = apr,
  number    = {4},
  volume    = {63},
  doi       = {10.1063/5.0087116},
  publisher = {AIP Publishing},
  url       = {http://dx.doi.org/10.1063/5.0087116},
}

@Article{BastiCenGiuOlg24,
  author  = {Giulia Basti and Serena Cenatiempo and Alessandro Giuliani and Alessandro Olgiati and Giulio Pasqualetti and Benjamin Schlein},
  journal = {Arch. Ration. Mech. Anal.},
  title   = {Upper bound for the ground state energy of a dilute {B}ose gas of hard spheres},
  year    = {2024},
  volume  = {248},
  eprint  = {2212.04431},
}

@article{Bijl1940lowestWF,
	title = {The lowest wave function of the symmetrical many particles system},
	journal = {Physica},
	volume = {7},
	number = {9},
	pages = {869-886},
	year = {1940},
	issn = {0031-8914},
	doi = {https://doi.org/10.1016/0031-8914(40)90166-5},
	url = {https://www.sciencedirect.com/science/article/pii/0031891440901665},
	author = {A. Bijl},
}

@article{Chen2011quinticNLS,
	title = {The quintic {NLS} as the mean field limit of a boson gas with three-body interactions},
	journal = {J. Funct. Anal.},
	volume = {260},
	number = {4},
	pages = {959--997},
	year = {2011},
	issn = {0022-1236},
	doi = {https://doi.org/10.1016/j.jfa.2010.11.003},
	url = {https://www.sciencedirect.com/science/article/pii/S0022123610004374},
	author = {Thomas Chen and Nataša Pavlović}
}

@article{Chen2018TheDO,
	title={The derivation of the {$\mathbb{T}^{3}$} energy-critical {NLS} from quantum many-body dynamics},
	author={Xuwen Chen and Justin Holmer},
	journal={Invent. Math.},
	year={2018},
	pages={1-115},
	url={https://api.semanticscholar.org/CorpusID:119127967}
}

@article{Comparin2015Liquid,
	title={Liquid-solid transitions in the three-body hard-core model},
	volume={109},
	ISSN={1286-4854},
	url={http://dx.doi.org/10.1209/0295-5075/109/20003},
	DOI={10.1209/0295-5075/109/20003},
	number={2},
	journal={EPL},
	publisher={IOP Publishing},
	author={Comparin, Tommaso and Kapfer, Sebastian C. and Krauth, Werner},
	year={2015},
	month=jan,
	pages={20003}
}

@article{Dyson1957gseHS,
	author={Dyson, F. J.},
	title={Ground state energy of a hard-sphere gas},
	journal={Phys. Rev.},
	volume={106},
	year={1957},
	pages={20--26},
}

@article{Dingle1949ZPE,
	author={Dingle, R.},
	title={The zero-point energy of a system of particles},
	journal={London Edinburgh Philos. Mag. \& J. Sci.},
	volume={40},
	year={1949},
	number={304},
	pages={573--578}
}

@article{Jastrwo1955ManyBody,
	title = {Many-Body Problem with Strong Forces},
	author = {Jastrow, Robert},
	journal = {Phys. Rev.},
	volume = {98},
	issue = {5},
	pages = {1479--1484},
	numpages = {0},
	year = {1955},
	publisher = {American Physical Society},
	doi = {10.1103/PhysRev.98.1479},
	url = {https://link.aps.org/doi/10.1103/PhysRev.98.1479}
}

@misc{lee2021rateCT,
	title={Rate of convergence towards mean-field evolution for weakly interacting bosons with singular three-body interactions}, 
	author={Jinyeop Lee},
	year={2021},
	eprint={2006.13040},
	archivePrefix={arXiv},
	primaryClass={math-ph}
}

@article{Li2021derivationNS,
	author = {Li, Yongsheng and Yao, Fangyan},
	title = "{Derivation of the nonlinear {S}chrödinger equation with a general nonlinearity and {G}ross--{P}itaevskii hierarchy in one and two dimensions}",
	journal = {J. Math. Phys. },
	volume = {62},
	number = {2},
	pages = {021505},
	year = {2021},
	issn = {0022-2488},
	doi = {10.1063/5.0035676},
	url = {https://doi.org/10.1063/5.0035676},
	eprint = {https://pubs.aip.org/aip/jmp/article-pdf/doi/10.1063/5.0035676/16098811/021505\_1\_online.pdf},
}

@book{lieb_loss,
	author={Lieb, E. H. and Loss, M.},
	title={{A}nalysis},
	series={{G}raduate {S}tudies in {M}athematics},
	publisher={{A}merican {M}athematical {S}ociety, second ed.},
	volume={14},
	year={2001},
}

@article{Nam2023condensation,
	author={Nam, Phan T. and Ricaud, J. and Triay, A.},
	title={The condensation of a trapped dilute {B}ose gas with three-body interactions},
	journal={Prob. Math. Phys},
	volume = {4},
	pages={91--149},
	year={2023},
}

@article{Nam2022dilute,
	author={Nam, Phan T. and Ricaud, J. and Triay, A.},
	title={Dilute {B}ose gas with three-body interaction: recent results and open questions},
	journal={J. Math. Phys.},
	volume={63},
	number={6},
	year={2022},
}

@article{Nam2022ground,
	author={Nam, Phan T. and Ricaud, J. and Triay, A.},
	title={Ground state energy of the low density {B}ose gas with three-body interactions},
	journal={J. Math. Phys.},
	volume={63},
	year={2022},
	pages={071903},
	note={Special collection in honor of
	Freeman Dyson},
}

@article{Nguyen2023onedimensional,
	author = {Nguyen, Dinh-Thi and Ricaud, Julien},
	title = {On One-Dimensional {B}ose Gases with Two\nobreakdash-Body and (Critical) Attractive Three\nobreakdash-Body Interactions},
	fjournal = {SIAM Journal on Mathematical Analysis},
	journal = {SIAM J.~Math. Anal.},
	publisher = {SIAM},
	coden = {SJMAAH},
	issn = {0036-1410; 1095-7154},
	volume = {56},
	number = {3},
	pages = {3203--3251},
	year = {2024},
	doi = {10.1137/22M1535139},
	url = {https://doi.org/10.1137/22M1535139},
	hal_id = {hal-03809946},
	eprint = {arXiv:2210.04515},
	primaryclass = {math-ph},
	mscclass = {81V70,35J10,35Q55,82B10,82D05}
}

@article{Piatecki2014Efimov,
	title={{E}fimov-driven phase transitions of the unitary {B}ose gas},
	volume={5},
	ISSN={2041-1723},
	pages={3503},
	url={http://dx.doi.org/10.1038/ncomms4503},
	DOI={10.1038/ncomms4503},
	number={1},
	journal={Nat. Commun.},
	publisher={Springer Science and Business Media LLC},
	author={Piatecki, Swann and Krauth, Werner},
	year={2014},
	month=mar
}

@article{Yuan2015derivationQNLS,
	title = {Derivation of the Quintic {NLS} from many-body quantum dynamics in {$T^2$}},
	author = {Jianjun Yuan},
	journal = {Commun. Pure Appl. Anal.,},
	volume = {14},
	number = {5},
	pages = {1941--1960},
	year = {2015},
	issn = {1534-0392},
	doi = {10.3934/cpaa.2015.14.1941},
	url = {https://www.aimsciences.org/article/id/d9dab15a-00c7-4dc8-89bc-0cc4ebcaefb9},
}

@Article{Xie15,
  author  = {Zhihui Xie},
  journal = {Differ. Integral Equ.},
  title   = {Derivation of a nonlinear Schrödinger equation with a general power-type nonlinearity in $d=1,2$},
  year    = {2015},
  pages   = {455--504},
  volume  = {28},
}

@Article{NamSal19,
  author  = {Phan Thành Nam and Robert Salzmann},
  journal = {Commun. Math. Phys.},
  title   = {Derivation of {3D} energy-critical nonlinear {S}chr\"odinger equation and {B}ogoliubov excitations for {B}ose gases},
  year    = {2019},
  pages   = {495--571},
  volume  = {375},
}

@Article{Chen11,
  author    = {Chen, Xuwen},
  journal   = {Arch. Ration. Mech. Anal.},
  title     = {Second Order Corrections to Mean Field Evolution for Weakly Interacting Bosons in The Case of Three-body Interactions},
  year      = {2011},
  issn      = {1432-0673},
  number    = {2},
  pages     = {455--497},
  volume    = {203},
  doi       = {10.1007/s00205-011-0453-8},
  publisher = {Springer Science and Business Media LLC},
}

@Article{AdamiLee25,
  author  = {Riccardo Adami and Jinyeop Lee},
  journal = {J. Funct. Anal.},
  title   = {Microscopic derivation of a {S}chr\"odinger equation in dimension one with a nonlinear point interaction},
  year    = {2025},
  number  = {10},
  pages   = {110866},
  volume  = {288},
}

@Article{NguyenRic25,
  author  = {Dinh-Thi Nguyen and Julien Ricaud},
  journal = {Lett. Math. Phys.},
  title   = {Stabilization against collapse of {2D} attractive {B}ose--{E}instein condensates with repulsive, three\nobreakdash-body interactions},
  year    = {2025},
  number  = {31},
  volume  = {115},
}

@Article{RoutSoh23,
  author  = {Andrew Rout and Vedran Sohinger},
  journal = {Commun. Partial Differ. Equ.},
  title   = {A microscopic derivation of {G}ibbs measures for the {1D} focusing quintic nonlinear {S}chr\"{o}dinger equation},
  year    = {2023},
  pages   = {1008--1055},
  volume  = {48},
}

@Article{Visconti26,
  author  = {Fran\c{c}ois L. A. Visconti},
  journal = {Ann. Henri Poincaré},
  title   = {Ground state energy of the low density Bose gas with two-body and three-body interactions},
  year    = {2026},
}

\end{document}